\documentclass[10pt, conference, letterpaper]{IEEEtran}
\usepackage{lscape}
\usepackage{amsthm}
\usepackage{amssymb}
\usepackage{amsmath}
\usepackage{multirow}
\usepackage{graphicx}
\usepackage{booktabs} 
\usepackage{tabularx}
\graphicspath{{figs/}}
\usepackage[font=small]{caption}
\usepackage[labelformat=simple,font=scriptsize]{subcaption}

\usepackage[linesnumbered,ruled,vlined,noend]{algorithm2e}
\makeatletter
\newcommand{\nosemic}{\renewcommand{\@endalgocfline}{\relax}}
\newcommand{\dosemic}{\renewcommand{\@endalgocfline}{\algocf@endline}}
\let\oldnl\nl
\newcommand{\nonl}{\renewcommand{\nl}{\let\nl\oldnl}}
\makeatother

\makeatletter
\newcounter{problemenv}

\makeatother
\usepackage{enumitem}
\usepackage[dvipsnames]{xcolor}
\usepackage{cite}
\usepackage[hidelinks]{hyperref}
\usepackage{lastpage}
\usepackage{setspace}
\usepackage{relsize}
\usepackage{makecell}

\newtheorem{theorem}{Theorem}

\newcommand{\eg}{\textit{e.g.},}
\newcommand{\ie}{\textit{i.e.},}

\newcommand{\eqend}{\,.}

\newcommand{\cat}[1]{\medskip\noindent\textbf{#1.}}

\newcommand{\SysName}{\textsmaller{\textsf{{dSamp}}}}

\newcommand{\ILP}{{\textsmaller {\textsf{ILP}}}}
\newcommand{\APX}{{\textsmaller {{\textsf{APX}}}}}

\newcommand{\csamphundred}{{\textsmaller{\textsf{cS+100}}}}
\newcommand{\csamphundredfifty}{{\textsmaller{\textsf{cS+150}}}}
\newcommand{\csamptwohundred}{{\textsmaller{\textsf{cS+200}}}}

\newcommand{\ISOCP}{{\textsmaller{\textsf{ISOCP}}}}
\newcommand{\DS}{{\textsmaller{{\textsf{DS}}}}}
\newcommand{\DStwo}{{\textsmaller {{\textsf{DS+2$\sigma$}}}}}

\newcommand{\dSamp}{{\textsmaller {{\textsf{dSamp}}}}}

\newcommand{\pr}[1]{\mathbb{P}\left\{#1\right\}}

\newcommand{\mc}[1]{\mathcal{#1}}

\begin{document}

\title{\scalebox{0.9}{Coordinated Sampling in SDNs with Dynamic Flow Rates}}

\author{
	\IEEEauthorblockN{
		Soroosh~Esmaeilian\IEEEauthorrefmark{1},
		Mahdi~Dolati\IEEEauthorrefmark{2},
		Sogand~Sadrhaghighi\IEEEauthorrefmark{1},
		and
		Majid~Ghaderi\IEEEauthorrefmark{1}
	}
	\IEEEauthorblockA{
		\IEEEauthorrefmark{1}Dept. of Computer Science, University of Calgary, Calgary, Canada,
		\IEEEauthorrefmark{2}School of Computer Science, IPM, Tehran, Iran.
	}
}

\maketitle

\begin{abstract}
Traffic sampling has become an indispensable tool in network management. While there exists a plethora of sampling systems, they generally assume flow rates are stable and predictable over a sampling period. Consequently, when deployed in networks with dynamic flow rates, some flows may be missed or under-sampled, while others are over-sampled. This paper presents the design and
evaluation of \SysName, a network-wide sampling system capable of handling dynamic flow rates in Software-Defined Networks (SDNs). The key idea in \SysName\ is to consider flow rate fluctuations when deciding on which network switches and at what rate to sample each flow. To this end, we develop a general model for sampling allocation with dynamic flow rates, and then design an efficient approximate integer linear program called \APX\ that can be used to compute sampling allocations even in large-scale networks.
To show the efficacy of \SysName\ for network monitoring, we have implemented \APX\ and several existing solutions in ns-3 and conducted extensive experiments using model-driven as well as trace-driven simulations. Our results indicate that, by considering dynamic flow rates, \APX\ outperforms the existing solutions by up to $10\%$ in sampling more flows at a given sampling rate.
\end{abstract}

\section{Introduction}
\label{s:intro}
\cat{Motivation}
Network traffic monitoring is crucial for many network
management tasks such as traffic engineering \cite{tsai2018network}, anomaly
detection~\cite{li2011scan}, traffic matrix estimation\cite{tian2018sdn} and customer accounting~\cite{duffield2001charging}. While aggregate
traffic statistics (\eg\ average traffic rate) are sufficient for some network management tasks (\eg\ traffic engineering), there is a growing need for having detailed flow-level measurements in other tasks, such as traffic classification~\cite{zhang2013unsupervised}. 
To address this need, many network operators employ monitoring solutions based on packet sampling, where network switches sample only a subset of the packets that pass through them~\cite{cohen2018sampling}, \cite{li2016flowradar}. However, sampling is a resource-intensive operation
for network switches, given their restricted processing capabilities.
A typical switch can only sample a tiny
fraction of the packets it encounters before its performance
starts to degrade~\cite{suh2014opensample}. 

Sampling can be performed on a per-port or per-flow basis. In per-port sampling solutions,
such as NetFlow~\cite{claise2004cisco} and sFlow~\cite{sflow}, a sampling rate is specified
for each switch input port. As a result, flows with higher
packet rates are more likely to be sampled compared to flows
with lower rates. Consequently, these solutions are inadequate for network management tasks that require a minimum per-flow sampling
rate, \eg\ anomaly detection~\cite{cantieni2006reformulating}.
Per-flow sampling addresses this bias toward high-rate flows by specifying a target sampling rate for each \textit{flow}. Several works such as~\cite{suarez2017towards} and~\cite{shirali2013flexam} have proposed extensions to the standard OpenFlow to enable per-flow sampling on commodity OpenFlow switches. Sampling solutions based on per-flow sampling are proposed in~\cite{shirali2013flexam} and~\cite{raspall2012efficient}. However, in these solutions, switches sample flows independently, resulting in redundant flow samples and inefficient use of switch resources. 

To prevent redundant sampling and enhance the flow monitoring capabilities of the network, coordinated sampling is considered~\cite{sekar2008csamp,sekar2010revisiting,cohen2018sampling,sadrhaghighi2022flowshark}, where a centralized controller coordinates monitoring responsibilities across different switches in the network. 
To determine per-flow sampling rates on each switch, these works generally assume that the flow rates are deterministic in the sense that they are 1) fixed over time, and 2) known a priori. In reality, flow rates are variable and fluctuate over time. As a result, these solutions may lead to under-sampling or over-sampling when the actual traffic rate of a flow deviates from the one assumed by the solution. This weakness renders such solutions unsuitable for many network management tasks that require a consistent minimum flow sampling rate to ensure effective and accurate performance.

As a motivating example, consider the small network depicted in Fig.~\ref{toy_example}. The network consists of two switches denoted by $S_1$ and $S_2$, each having a sampling capacity of $3$~packets per second (pps). There are four traffic flows denoted by $f_1$,$f_2$,$f_3$ and $f_4$ that are required to be sampled at the sampling rate of $0.1$. The rates of $f_1$ and $f_2$ follow a Normal distribution with a mean of $5$~pps and a standard deviation of $10$~pps, while the rates of $f_3$ and $f_4$ follow a Normal distribution with a mean of $14$~pps and standard deviation of $1$~pps.
A deterministic sampling solution that ignores rate fluctuations assigns flows $f_1$ and $f_3$ to switch $S_1$, and flows $f_2$ and $f_4$ to switch $S_2$ for sampling. This assignment ensures equal sampling load distribution (with respect to the mean flow rates) between the switches, which is a common objective in coordinated sampling solutions~\cite{sekar2008csamp,sadrhaghighi2022flowshark}. Let $\delta_i$ denote the probability that the capacity of switch $i$ is violated (\ie\ its sampling capacity is not sufficient to sample its flows at the specified target sampling rate). We have:
\[
	\delta_1 = \delta_2 = \pr{0.1 \times \mathcal{N}(14+5, 1+100) > 3} = 0.14
	\eqend
\]
Next, consider a sampling solution that considers variability in flow rates, such as the one proposed in this paper. It may decide on the following assignment, where flows $f_1$ and $f_2$ are assigned to switch $S_1$ and $f_3$ and $f_4$ are assigned to switch $S_2$. In this case, we obtain that:
\begin{align*}
	\delta_1 &= \pr{0.1 \times \mathcal{N}(5+5, 100+100) > 3} = 0.08,\\
	\delta_2 &= \pr{0.1 \times \mathcal{N}(14+14, 1+1) > 3} = 0.08,
\end{align*}
which represents almost $43\%$ reduction in the switch capacity violation probability by simply considering fluctuations in flow rates. Any reduction in the switch capacity violation probability translates to more precise flow sampling and, consequently, enhanced visibility of network flows especially for management tasks that require a minimum sampling rate.

\begin{figure}[t]
	\centering
	\includegraphics[width=80mm]{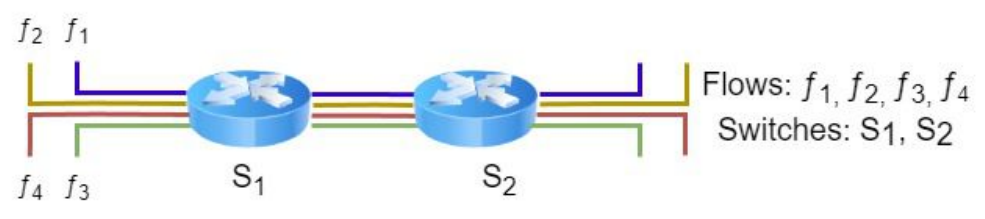}
	\caption{A small network with two switches and four flows.}
	\label{toy_example}
	\vspace{-5mm}
\end{figure}

\cat{Our work}
We introduce (\SysName), a network-wide sampling solution capable of handling dynamic flow rates. \SysName\ is designed to operate with commodity OpenFlow switches and SDN network controllers. The key idea behind \SysName\ design is to incorporate statistical information about flow rates into the process of making sampling decisions. This is in contrast to existing works that assume knowledge about the \textit{exact rate} of flows. Specifically, \SysName\ only requires information about the mean and variance of flow rates. As discussed earlier, it is unrealistic to assume that flow rates are fixed and known exactly in advance. However, not only it is straightforward to obtain statistical information about flows, for example, based on historical measurements or service level agreements, but also this information is more robust (\ie\ does not fluctuate significantly) compared to the instantaneous flow rates. 
 
Our main contributions in this paper are:
\begin{itemize}
	\item We present the design and evaluation of \SysName, a network-wide sampling system, capable of handling flow rate fluctuations in SDNs. 
	\item  We formulate flow sampling with statistical flow rate information as an Integer Second Order Cone Program (ISOCP). Given the complexity of solving the ISOCP formulation for both small and large-scale networks, we present two reformulations to obtain solutions with practical runtimes. The first reformulation is a transformation and more suitable for small size networks, while the second reformulation is approximate but suitable for larger networks.
	\item We provide extensive model-driven and trace-driven simulations using ns-3. For trace-driven simulations, we use real traffic traces collected from an ISP. Our results show that \SysName\ outperforms existing sampling solutions based on deterministic flow rates. 
\end{itemize}

\cat{Organization}
Related works are reviewed in Section~\ref{s:related}. The high-level design of \SysName\ is presented
in Section~\ref{s:system_design}. In Section~\ref{s:optimization}, we present our algorithms and
their analysis. Evaluation results are presented in Section ~\ref{s:eval}. The concluding remarks are presented in Section ~\ref{s:conc}.

\begin{figure}[ht!]
	\centering
	\includegraphics[width=80mm]{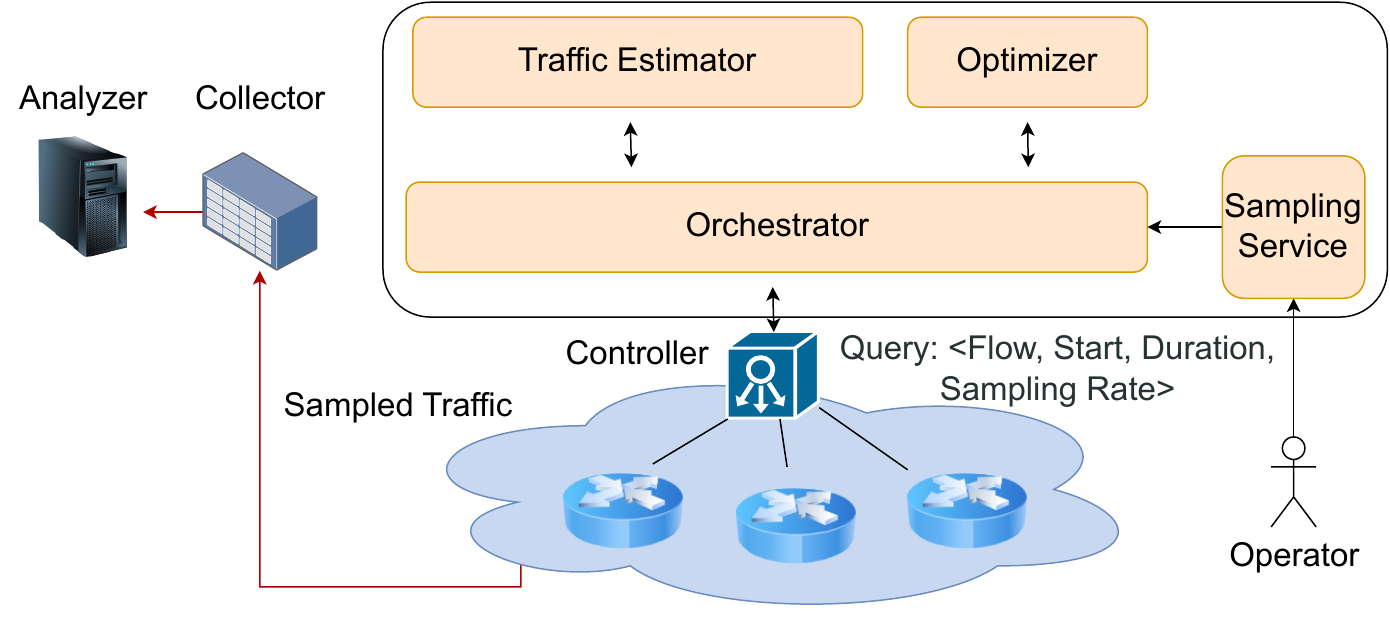}
	\caption{High-level architecture of \SysName. \label{design}}
\end{figure}

\section{System Design}\label{s:system_design}
The high-level architecture of \SysName\ is depicted in Fig.~\ref{design}. The main component of \SysName\ is the orchestrator, which is implemented on top of the SDN controller. The design of \SysName\ also includes an optimizer module that works over sampling epochs. The sampling queries that are sent to the sampling service during the current sampling epoch are batched together and then queried by the orchestrator. Afterwards, the orchestrator processes them in a single operation at the beginning of the following sampling epoch. The epoch-based structure enables \SysName\ to adapt its sampling decisions to the dynamic set of flows at the start of each epoch. In our design, a flow is specified by a source-destination IP address pair. Thus, a flow may represent the aggregated traffic between an Origin-Destination (OD) pair in an ISP network or a pair of virtual machines in a datacenter network. The length of the sampling epoch is a design parameter that can be tuned to achieve different trade-offs. In particular, a shorter epoch allows the system to adapt more quickly to changes in network traffic, while a longer epoch helps reduce the control overhead. In the following, we provide a high-level overview of the system's operational workflow, followed by a detailed examination of the individual components within the architecture.

\cat{System Workflow} 
There are two phases in \SysName\ operation, namely the setup phase and the sampling phase.
In the setup phase, the network operator adjusts various system parameters such as the sampling epoch length.

Once the setup phase is completed, the sampling phase starts. As stated before, \SysName\ works over sampling epochs. At the beginning of each epoch, the orchestrator queries the controller to collect statistical data about flows and their respective paths. Additionally, it queries the sampling service to identify the flows to be sampled. It then invokes traffic estimator and optimizer modules to compute a sampling schedule for the selected flows. A \textit{sampling schedule} determines which switch in the network will sample each flow. Finally, the orchestrator communicates with the controller to install the newly computed sampling schedule on the relevant switches. The destination of all sampled traffic is the collector server (essentially a load balancer in front of storage servers), which is linked to analysis applications that carry out the required management task.

\cat{Sampling Service} The sampling service receives sampling queries from the users (\eg\ network operators and tenants). At the end of each epoch, the orchestrator collects all the sampling queries from the sampling service. Queries are defined by a tuple $\mathrm{<flow, start, duration, sampling~rate>}$. 
The filed ``start'' denotes the start time of sampling, while ``duration'' indicates the duration of sampling for the specified flow. The filed ``sampling rate'' determines the desired sampling rate.

\cat{Orchestrator} The orchestrator is the main component of \SysName\ and includes two modules, namely traffic estimator and optimizer. All communications between \SysName\ and network switches go through the orchestrator. At the start of each epoch, the orchestrator collects all sampling queries from the sampling service. Additionally, it obtains flow statistics and their paths from the controller for the recently ended epoch. The orchestrator passes this information to the traffic estimator and, in turn, receives the mean and variance for each flow's rate. Using the mean and variance of each flow's rate, as well as their paths, the orchestrator leverages the optimizer module to create a sampling schedule for each flow. This schedule is then installed on network switches via the controller.

\cat{Traffic Estimator} To calculate an optimal sampling schedule for flows with non-deterministic rates, \SysName\ requires the mean and variance estimates of the flow rate. To this extent, various approaches can be utilized, \eg~\cite{zhou2006traffic} employs an autoregressive model, while\cite{andreoletti2019network} uses machine learning to estimate flow rates. In our design, we applied a simple estimator based on past traffic statistics. Specifically, the traffic estimator queries the controller to collect statistics about flow rates over the past epochs and then estimates the mean and variance of flow rates for the upcoming epoch based on the collected statistics. 

\cat{Optimizer} The optimizer is responsible for computing a sampling schedule while considering flow rate fluctuations. The goal is to optimize the use of switch sampling resources in order to maximize the number of fully sampled flows, \ie\ flows that are sampled at their specified target rate. To achieve this goal, \SysName\ limits sampling each flow to a single switch along its path. Hence, depending on the available switch capacity, a flow either could be sampled on one switch along its path or will not be sampled at all. This will satisfy two goals: 1)  avoiding duplicate sampled packets, and 2) improving flow visibility, as \SysName\ only samples flows for which there is a high confidence of fully sampling them. This way, \SysName\ gives higher priority to fully sampling flows (for improved flow visibility) rather than sampling many flows but only partially. This design is advantageous for tasks requiring a specific minimum sampling rate for each individual flow. The primary focus of the following section is the design of efficient algorithms for the optimizer. 

\begin{table}[h]
	\centering
	\fontsize{8}{8}\selectfont
	\renewcommand{\arraystretch}{1.05}
	\caption{Table of Notations.}
	\label{t:table_notations}
	\begin{tabular}{m{0.15\columnwidth}m{0.72\columnwidth}} 
		\Xhline{\arrayrulewidth}
		\Xhline{0.1pt}
		\\[-0.8em]
		\textbf{Notation} & \textbf{Description} \\ [0.5ex] 
		\hline
		\\[-0.8em]
		${r}_{f}$ & Sampling load of flow $f$ \\ 
		${\mu}_{f}$ & Mean sampling load for flow $f$ \\
		$\sigma^2_{f}$ & Variance of sampling load for flow $f$ \\
		${x}_{f,s}$ & Sampling decision for flow $f$ on switch $s$  \\
		${B}_{s}$ & Sampling capacity of switch  $s$  \\
		$\mathcal{F}$ & Set of flows to be sampled \\
		$\mathcal{S}$ & Set of switches in the network \\
		$\mathcal{F}_{s}$ & Set of flows that pass through switch $s$\\
		$\mathcal{S}_{f}$ & Set of switches on the path of flow $f$\\
		$\delta$ & Probability of violating a switch sampling capacity \\
		$\alpha$ & Pre-specific target sampling rate \\
		[0.5ex] 
		\Xhline{\arrayrulewidth}
		\Xhline{0.1pt}  
	\end{tabular}
\end{table}
\cat{Discussion} While the design of \SysName\ primarily targets ISPs, it could also be altered to be used in DCs. In the DC setting, the sampling query specifies a list of VM pairs belonging to a tenant. Similar to the case of ISP deployment, the system reconfigures switches only at the beginning of each sampling epoch. This strategy is consistent with the original design, minimizing control overhead associated with collecting flow statistics and reconfiguring switches.
\section{Sampling Optimization}\label{s:optimization}
Our goal in this section is to design an \textit{efficient} algorithm for computing the network-wide sampling allocations with dynamic flow rates. The algorithm must be efficient in terms of monitoring performance as well as computational complexity so it can be applied to large network scenarios. To this end, we show that dynamic flow rates can be incorporated in the problem formulation, but they lead to second-order constraints that are computationally intensive to solve, even for small network scenarios. Consequently, we focus on developing an approximate formulation that can be utilized in large-scale networks, while achieving performance that is close to that of the exact formulation. The important notions used in this section are shown in Table~\ref{t:table_notations}.

\subsection{Problem Formulation}
Let $\mc{F}$ denote the set of \emph{flows} to be sampled in the network. We define $\mc{S}$ to be the set of all switches in the network and $\mc{S}_{f}$ to be the set of switches on the path of flow $f\in\mc{F}$. For notation convenience, we also define $\mc{F}_{s}$ to be a subset of flows in $\mc{F}$ that pass through switch $s\in\mc{S}$. 
The network orchestrator requires each flow in $\mc{F}$ to be sampled at a target sampling rate (to achieve a network monitoring objective). Consequently, the sampling load of each flow, defined as the rate of sampled traffic transmitted to the collector, is computed from its rate and its target sampling rate. Each switch has a limited capacity to send sampled packets to the collector. We use $B_{s}$ to denote the transmission capacity of switch $s\in\mc{S}$ for sending the sampled packets. 
We consider a scenario where the flow rates are dynamic and fluctuate over time. Thus, we use the random variable $r_{f}$ to show the sampling load of flow $f$. This random variable models the fluctuation of the flow's rate, and any realization of $r_{f}$ is valid.

To formulate the sampling allocation problem, we define binary decision variable $x_{f,s}$, which is set to one if flow $f$ is sampled on switch $s$, and to zero otherwise. To avoid duplications, where a packet is sampled redundantly on multiple switches, we require that each flow be sampled either on one switch or nowhere at all if the available sampling capacity is not sufficient to handle the resulting sampling load. Consequently, we introduce the following constraint to ensure that each flow is sampled by at most one switch in the network:
\begin{gather} \label{eq_admit}
	\textstyle\sum_{s\in\mc{S}_{f}} x_{f,s} \le 1,
	\quad
	\forall f\in\mc{F}
	\eqend
\end{gather}
Given the stochastic nature of the flow rates, we use the following probabilistic constraints to capture the limited sampling capacity for each switch:
\begin{gather} \label{eq_bw_pr}
	\mathbb{P}\big\{
		\textstyle\sum_{f\in\mc{F}_{s}} r_{f}\cdot x_{f,s} \le B_{s}
	\big\} \ge 1-\delta,
	\quad 
	s\in\mc{S},
\end{gather}

where, $0 < \delta < 1$ is a small probability with which the network orchestrator accepts the violation of the sampling capacity of a switch. Note that the violation of the sampling capacity results in the loss of a subset of sampled packets, which translates into a lower realized sampling rate for the affected flows on that switch. Although a lower value for $\delta$ reduces the capacity violation probability, it also leads to rejecting more flows (\ie\ setting more $x_{f,s}$ to zero) to ensure that the available capacity is sufficient to absorb any fluctuation in the rate of those flows that are accepted for sampling (\ie\ flows whose corresponding $x_{f,s}$ is set to one).

Incorporating~\eqref{eq_bw_pr} directly into an optimization program results in a stochastic optimization, which is generally non-trivial to solve efficiently unless the exact distribution of variables $r_f$ is known, and can be expressed via simple closed-form expressions, none of which are true in real-world settings. Instead, a common approach in the literature\cite{sekar2010coordinated,cohen2018sampling,sadrhaghighi2022flowshark} is to consider a deterministic version of~\eqref{eq_bw_pr}, as follows:
\begin{gather}
	\textstyle\sum_{f\in\mc{F}_{s}} r_{f}\cdot x_{f,s} \le B_{s},
	\quad 
	s\in\mc{S},
\end{gather}
which is a linear constraint, but it entirely ignores the flow rate dynamics. However, a closer inspection of constraints~\eqref{eq_bw_pr} reveals that we are only interested in computing \textit{the tail probability of the sum of random variables}. There are many concentration bounds, such as Chernoff bounds\cite{chernoff1952measure}, that can be used to estimate such tail probabilities with different degrees of accuracy and complexity. If the target sampling rates are small, which is indeed the case in network-wide sampling, then each switch will have sufficient capacity to sample multiple flows. In such a setting, we can apply the central limit theorem to provide a tight estimation for the tail probability in~\eqref{eq_bw_pr}, which yields the following constraints: 
\begin{gather}\label{eq_bw_norm}
	\textstyle\sum_{f\in\mc{F}_{s}} \mu_{f}\cdot x_{f,s} 
	+
	z_{\delta}\textstyle\sqrt{
		\sum_{f\in\mc{F}_{s}} \sigma_{f}^{2}\cdot x_{f,s} 
	} \le B_{s},
\end{gather}
where, $z_{\delta}$ denotes the $(1-\delta)$ quantile of the standard Normal distribution. Putting this all together, the network-wide sampling allocation problem with dynamic flow rates can be defined as the following integer second order cone program:
\begin{align}
	\text{($\ISOCP$):} \qquad \text{max. } &\sum_{f\in\mathcal{F}}\sum_{s\in\mathcal{S}} x_{f,s}, \label{eq_cone_prob} \\
	\text{s.t. } &\eqref{eq_admit}, \eqref{eq_bw_norm}\eqend \notag
\end{align}

\vspace{-5mm}
\subsection{Linearization}
Given the conic constraints~\eqref{eq_bw_norm} in the ISOCP formulation, it is generally infeasible to solve with off-the-shelf solvers in a reasonable timeframe for sampling in realistic-sized networks. However, it is possible to transform these constraints into linear ones by introducing additional auxiliary integer variables in the formulation. Our numerical experiments indicated that this transformation significantly improves the computation time when employing off-the-shelf solvers such as Gurobi~\cite{GurobiOptimization}. To this end, we start by rewriting constraints~\eqref{eq_bw_norm} as follows where both sides of the inequality are squared:
\begin{align} \label{eq_bw_pow_2}
	z_{\delta}^{2}\textstyle\sum_{f\in\mc{F}_{s}} \sigma_{f}^{2}\cdot x_{f,s} 
	\le 
	(B_{s} - \textstyle\sum_{f\in\mc{F}_{s}} \mu_{f}\cdot x_{f,s})^{2}\eqend
\end{align}
To ensure the original constraints are enforced after this transformation, it is necessary that both sides of the inequality are non-negative. Clearly, the left-hand side is always positive, while the right-hand side is unrestricted. To resolve this, we also include the following constraints in the formulation:
\begin{align} \label{eq_pos}
	B_{s} - \textstyle\sum_{f\in\mc{F}_{s}} \mu_{f}\cdot x_{f,s} \ge 0\eqend
\end{align}
Then, we expand the right-hand side of~\eqref{eq_bw_pow_2} to obtain:
\begin{align}
	B_{s}^{2} - 2B_{s}\textstyle\sum_{f\in\mc{F}_{s}} \mu_{f}\cdot x_{f,s} + (\sum_{f\in\mc{F}_{s}} \mu_{f}\cdot x_{f,s})^{2},
\end{align}
where, the last term can be further expanded to yield the following expression:
\begin{align}
	\textstyle\sum_{f\in\mc{F}_{s}} \mu_{f}^{2}\cdot x_{f,s}^{2} 
	+ 
	2\sum_{f,f^{'}\in\mc{F}_{s}} \mu_{f}\cdot x_{f,s}\cdot \mu_{f^{'}}\cdot x_{f^{'},s}\eqend
\end{align}
Now, we introduce auxiliary variables $w_{f,f^{'},s}$ to model the multiplication of binary variables $x_{f,s}$ and $x_{f^{'},s}$. The following constrains ensure that $w_{f,f^{'},s}$ is only equal to one when both $x_{f,s}$ and $x_{f^{'},s}$ are equal to one:
\begin{align}
	&w_{f,f^{'},s} \le x_{f,s}, \label{eq_w_1} \\
	&w_{f,f^{'},s} \le x_{f^{'},s}, \label{eq_w_2} \\
	&w_{f,f^{'},s} \ge x_{f^{'},s} + x_{f,s} - 1\eqend \label{eq_w_3}
\end{align}
Using $w_{f,f^{'},s}$, we can re-write~\eqref{eq_bw_pow_2} as follows:
\begin{multline} \label{eq_soc_lined}
	z_{\delta}^{2}\textstyle\sum_{f\in\mc{F}_{s}} \sigma_{f}^{2}\cdot x_{f,s} 
	\le B_{s}^{2} - 2B_{s}\sum_{f\in\mc{F}_{s}} \mu_{f}\cdot x_{f,s}\\ 
	+ \sum_{f\in\mc{F}_{s}} \mu_{f}^{2}\cdot x_{f,s}  + 2\textstyle\sum_{f,f^{'}\in\mc{F}_{s}} \mu_{f}\mu_{f^{'}}\cdot w_{f,f^{'},s},
\end{multline}
which has a linear structure. Therefore, the original ISOCP can be transformed into the following linear program:
\begin{align}
	\text{($\ILP$):} \qquad \text{max. } &\sum_{f\in\mathcal{F}}\sum_{s\in\mathcal{S}} x_{f,s}, \label{eq_int_prob} \\
	\text{s.t. } &\eqref{eq_admit}, \eqref{eq_pos}, \eqref{eq_w_1}, \eqref{eq_w_2}, \eqref{eq_w_3}, \eqref{eq_soc_lined}\eqend \notag
\end{align}

\begin{theorem}
	\ILP\ requires at most $O(|\mc{F}|^{2}\mc{S})$ more decision variables compared with \ISOCP.
\end{theorem}
\begin{proof}
	\ILP\ requires decisions variables $x_{f,s}$ for constraints~\eqref{eq_admit} and decision variables $w_{f,f^{'},s}$ for the linearization process. The number of former variables is equal to the number of variables in \ISOCP. However, the latter variables are new and there are at most $O(|\mc{F}|^{2}\mc{S})$ of them. 
\end{proof}

\subsection{Approximation}
While experimenting with ILP, we found that even this formulation is computationally too expensive, and can be solved only for small-sized problem instances using Gurobi. Thus, we turned our attention to developing an approximation formulation that can be applied to large-scale problem instances. To this end, notice that the following inequality is always satisfied:
\begin{align}
	\sqrt{
		\textstyle\sum_{f\in\mc{F}_{s}} \sigma_{f}^{2}\cdot x_{f,s} 
	} \le \textstyle\sum_{f\in\mc{F}_{s}} \sigma_{f}\cdot x_{f,s}\eqend
\end{align}
Using this inequality, we can rewrite~\eqref{eq_bw_norm} as follows:
\begin{gather}\label{eq_bw_approx_lin}
	\textstyle\sum_{f\in\mc{F}_{s}} \mu_{f}\cdot x_{f,s} 
	+
	z_{\delta} \textstyle\sum_{f\in\mc{F}_{s}} \sigma_{f}\cdot x_{f,s}  \le B_{s},
\end{gather}
which is linear.
By substituting the above linear constraints in the original ISOCP, we obtain the following formulation:
\begin{align}
	\text{($\APX$):} \qquad \text{max. } &\sum_{f\in\mathcal{F}}\sum_{s\in\mathcal{S}} x_{f,s}, \label{eq_approx_prob} \\
	\text{s.t. } &\eqref{eq_admit}, \eqref{eq_bw_approx_lin}, \notag
\end{align}
which is an integer linear program with the \textit{same number of decision variables} as in the original formulation.

\section{Evaluations} \label{s:eval}
In this section, we evaluate the performance of \dSamp\ through extensive ns-3 simulations using both synthetic and real traffic traces. In our simulations, we have implemented a controller that calculates sampling allocations at the start of each epoch and configures the switches. When a packet arrives at a switch, it checks the switch configuration stored by the controller, which maps each flow's source and destination IP to a sampling rate. The switch identifies the packet's flow using its source and destination IP and samples it according to a probability equal to the sampling rate.
\subsection{Methodology}
We present three sets of ns-3 experiments. First, we start with a set of \textit{micro benchmarks}. These experiments are primarily designed to study the behavior of the approximate algorithm \APX\ in various settings in order to understand its sampling accuracy, computational efficiency, and sensitivity to various design assumptions, such as the distribution of traffic rates. Next, we present a set of \textit{model-driven} simulations to compare \APX\ with existing sampling solutions presented in\cite{sekar2008csamp} and\cite{cohen2018sampling}. Model-driven simulations provide us with the flexibility to simulate a range of network conditions without being restricted to a specific setup as in trace-driven simulations (which consider a traffic trace collected over a specific time period from one ISP). Specifically, we use model-driven experiments to fine-tune the baseline algorithms that are used for comparison, and use these fine-tuned versions in our trace-driven experiments. Finally, we present a set of \textit{trace-driven} simulations to demonstrate the ability of \APX\ to outperform existing sampling solutions using a real-world traffic trace collected from an ISP backbone\cite{MAWI}. 

\cat{Simulation Settings}
The model-driven and trace-driven simulations presented in this section were conducted in ns-3. The simulations were conducted on a standard desktop machine, equipped with an Intel(R) Core(TM) i9-12900 CPU @ 2.40 GHz and 16 GB of RAM. The \textit{default} settings are as follows:

\begin{itemize}[leftmargin=*]
\item \textit{Network:} The small-scale simulations utilize the Abilene topology with $11$ nodes and $14$ links. The link capacities are set to $10$~Gbps. The duration of each sampling epoch is set to $5$ seconds by default. The results are computed as the average of 5 independent simulation runs.  

\item \textit{Traffic:} The ns-3 on-off application was used to generate traffic based on the input traffic model. The packet size is set to $1$~KB in all simulations. For each flow, the source and destination are chosen at random. To describe the burstiness of traffic, we use the Coefficient of Variation (CoV), defined as CoV = $\mu/\sigma$. A high CoV indicates high variability in flow rates, whereas a low CoV suggests low variability. 

\item \textit{Sampling:} By default, the target sampling rate for each flow is set to $\alpha=0.1$, while the switch sampling capacity violation probability is set to $\delta=0.20$.

\end{itemize}

\begin{table}[t]
	\small
	\centering
	\caption{Runtime comparison between \APX\ and \ILP.}
	\label{tab:APXvsILP_admitted}
	\begin{tabular}{cccccc} 
		\toprule
		& \multicolumn{4}{c}{(Number of Flows, Switch Sampling Capacity)} & \\ 
		\cmidrule(lr){2-5}
		\multirow{-2}{*}{\hspace{0pt}} & (50, 70) & (100, 70) & (500, 70) & (100, 100) & \multirow{-2}{*}{} \\ 
		\midrule
		\APX\ & 0.001s & 0.003s & 0.005s & 0.004s & \\
		\ILP\ & 33.62s & 268.19s & \textgreater600s & \textgreater600s & \\
		\bottomrule
	\end{tabular}
\end{table}
\begin{figure}[t]
	\centering 
	\begin{subfigure}[t]{0.5\linewidth}
		\centering 
		\includegraphics[width=1\linewidth]{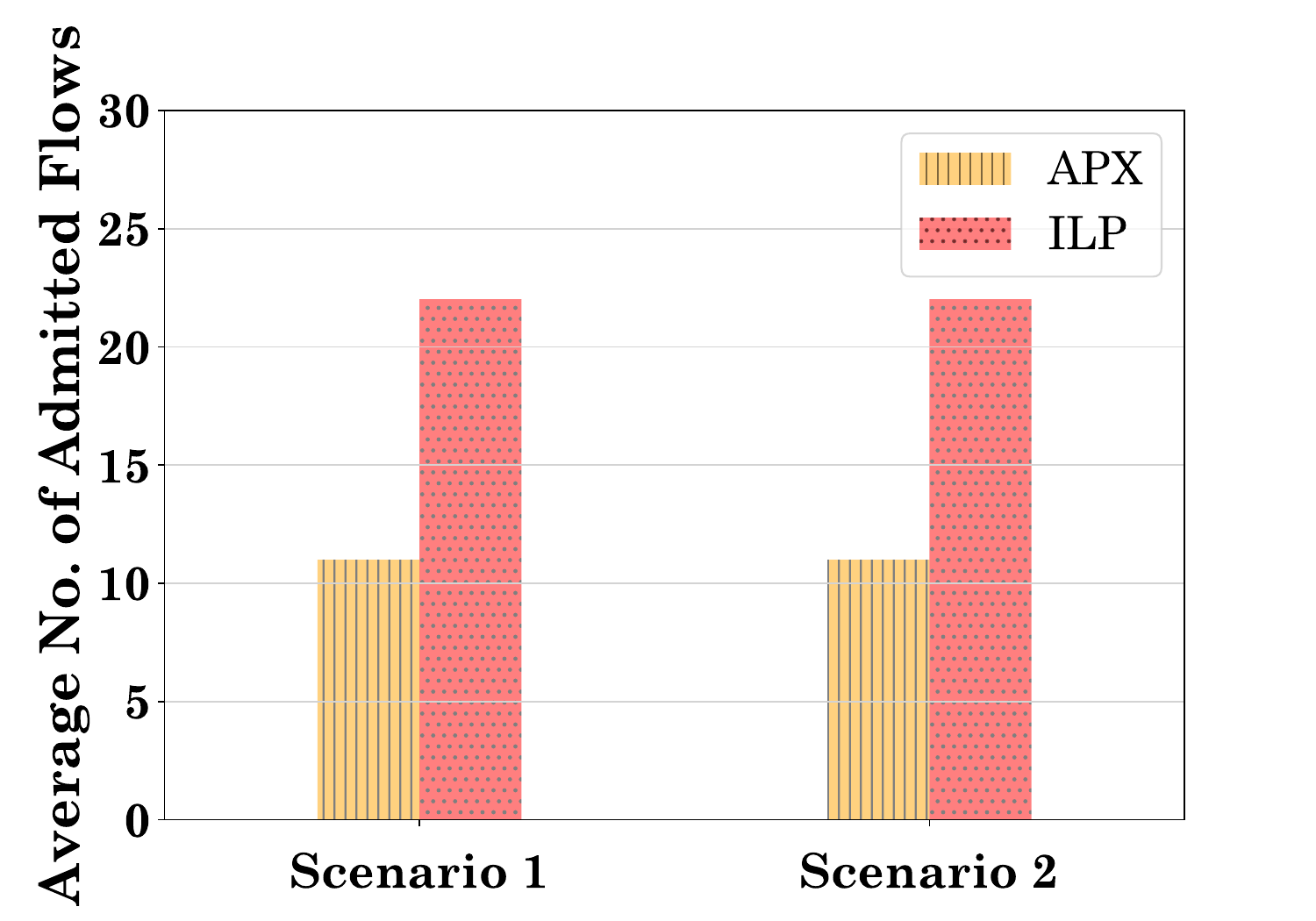} %
		\caption{Total number of admitted flows.}
		\label{APXvsILP_admitted}
	\end{subfigure}\hfil
	\begin{subfigure}[t]{0.5\linewidth}
		\centering
		\includegraphics[width=1\linewidth]{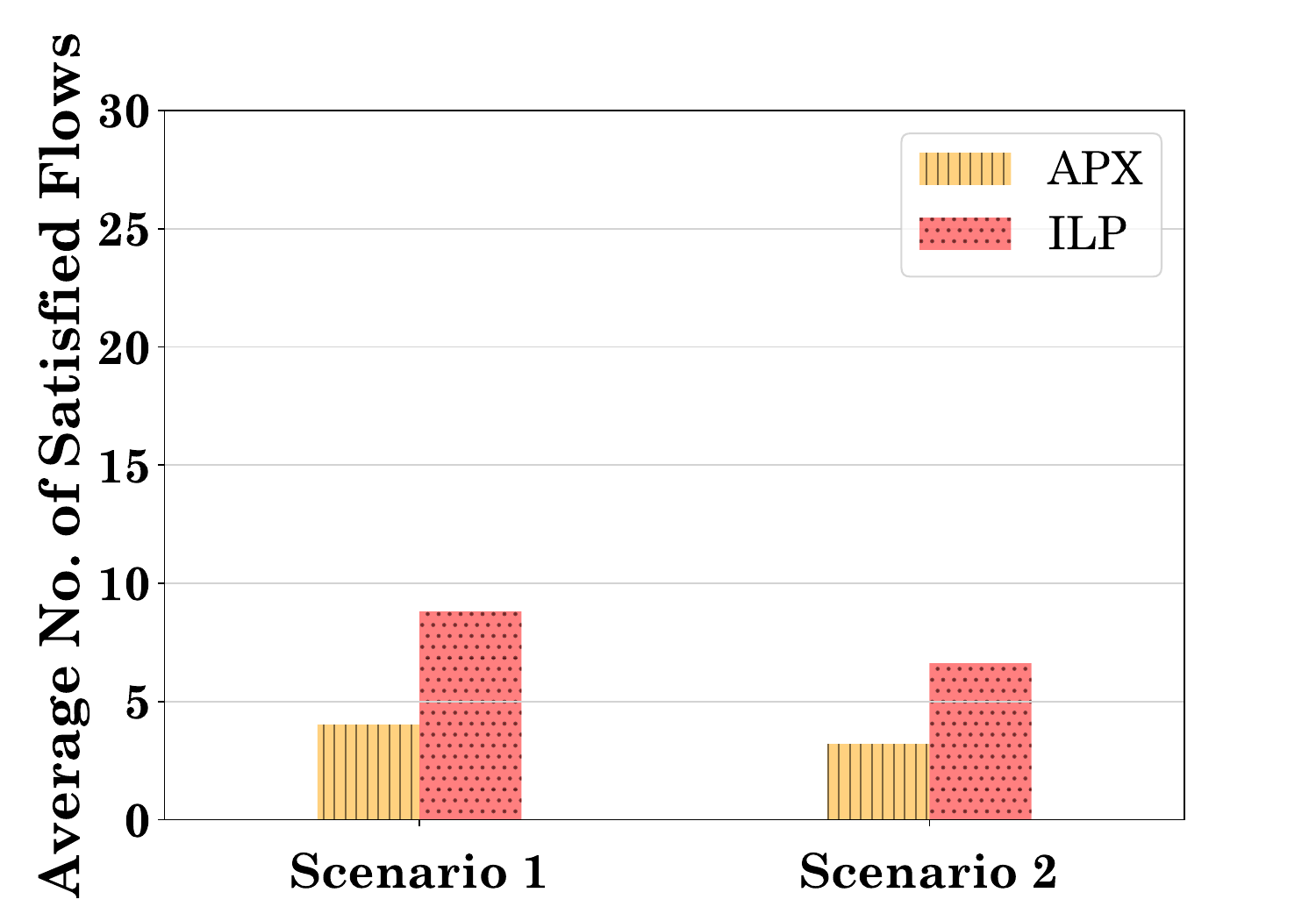}%
		\caption{Total number of fully sampled flows.}
		\label{APXvsILP_average}
	\end{subfigure}
	\caption{Performance comparison between \APX\ and \ILP.}%
	\label{fig:APXvsILP}%
	\vspace{-0.5cm}
\end{figure}

\begin{figure*}[t]
	\centering
	\begin{minipage}[t]{0.32\linewidth}
		\centering
		\includegraphics[width=\linewidth]{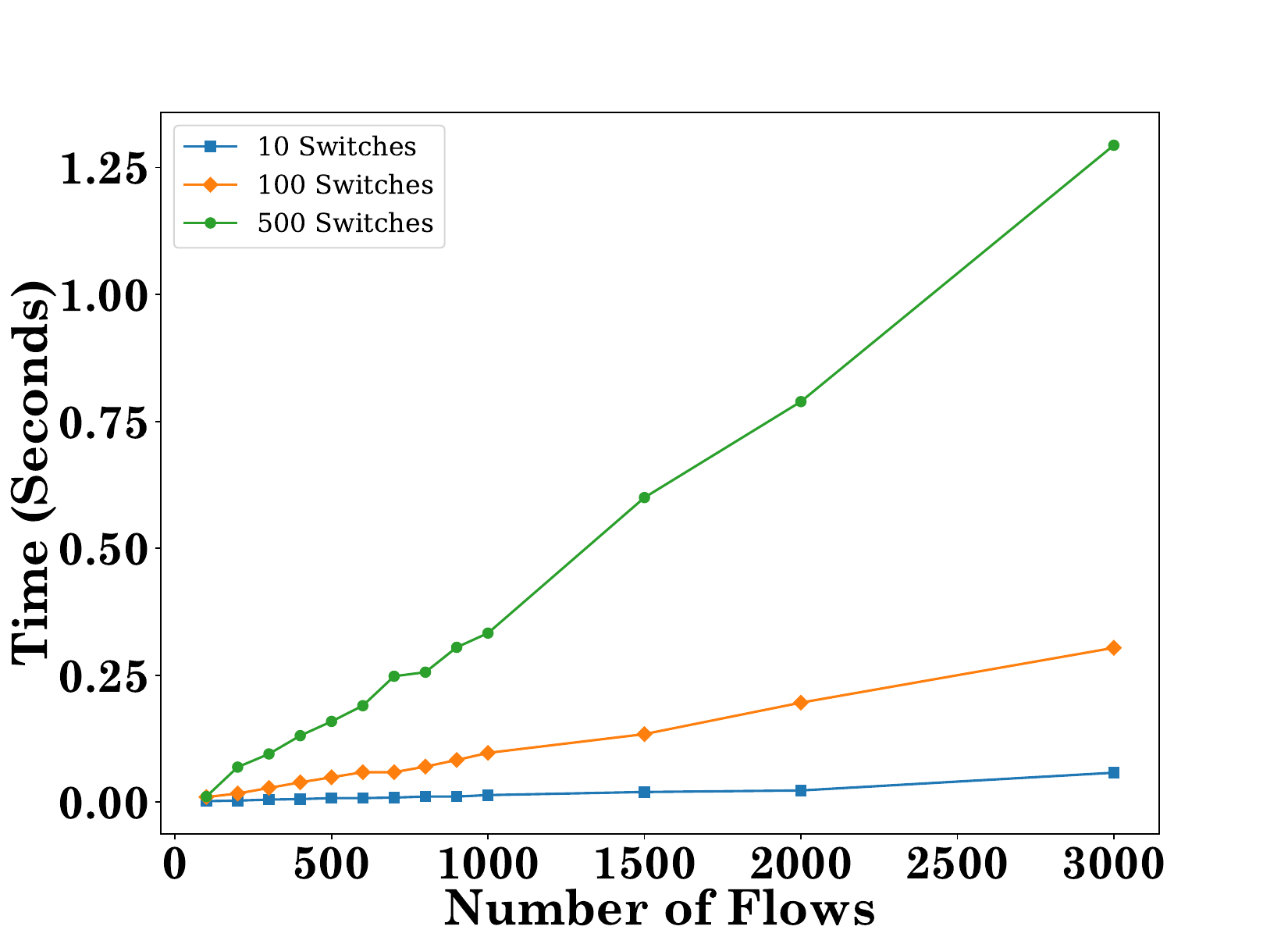}
		\caption{Measured runtime of \APX.}
		\label{APX_runtime}
	\end{minipage}
	\hfil
	\begin{minipage}[t]{0.33\linewidth}
		\includegraphics[width=\linewidth]{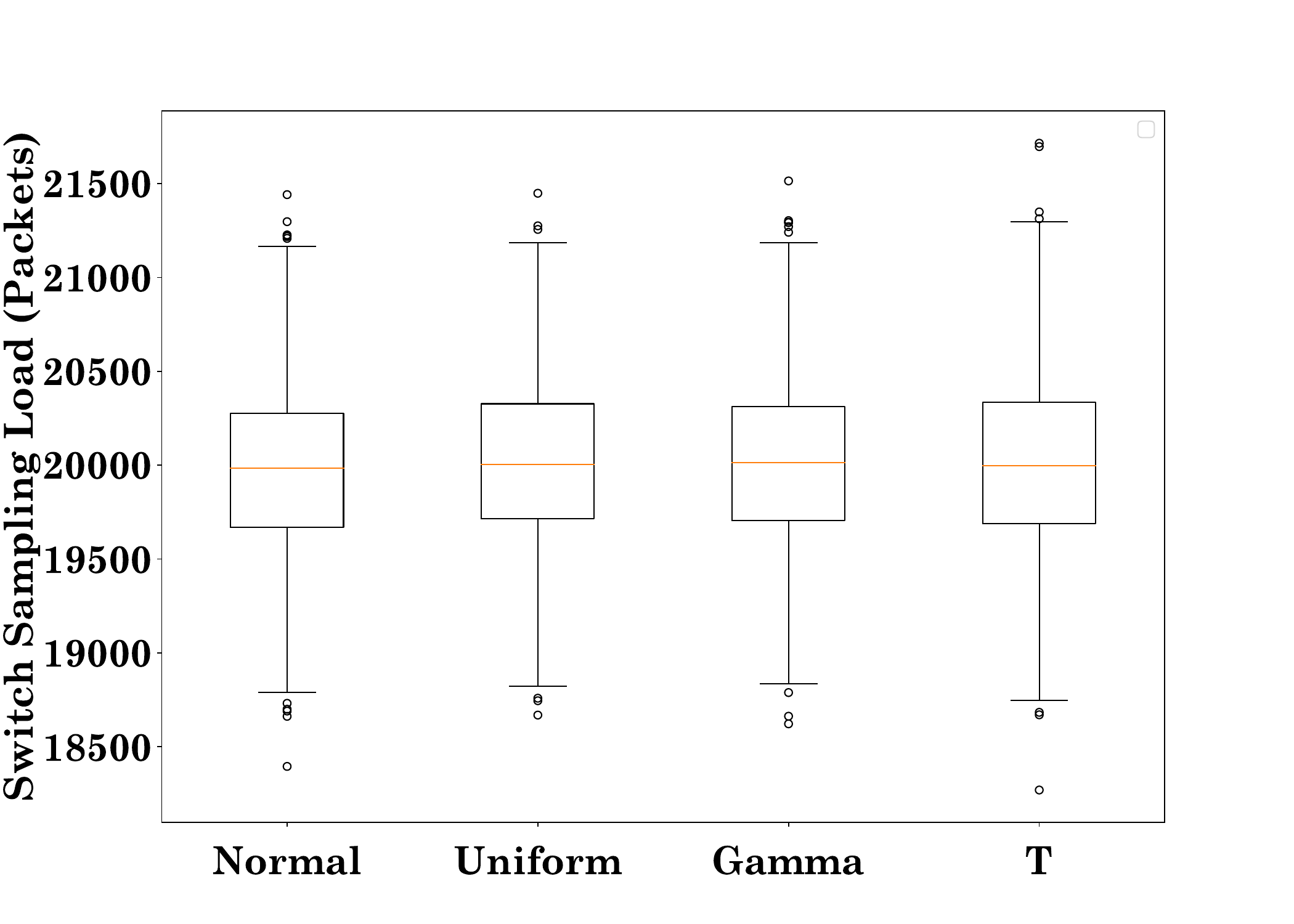}
		\caption{Effect of rate distribution on \APX.}
		\label{Distribution_analysis}
	\end{minipage}	
	\hfil
	\begin{minipage}[t]{0.32\linewidth}
		\centering
		\includegraphics[width=\linewidth]{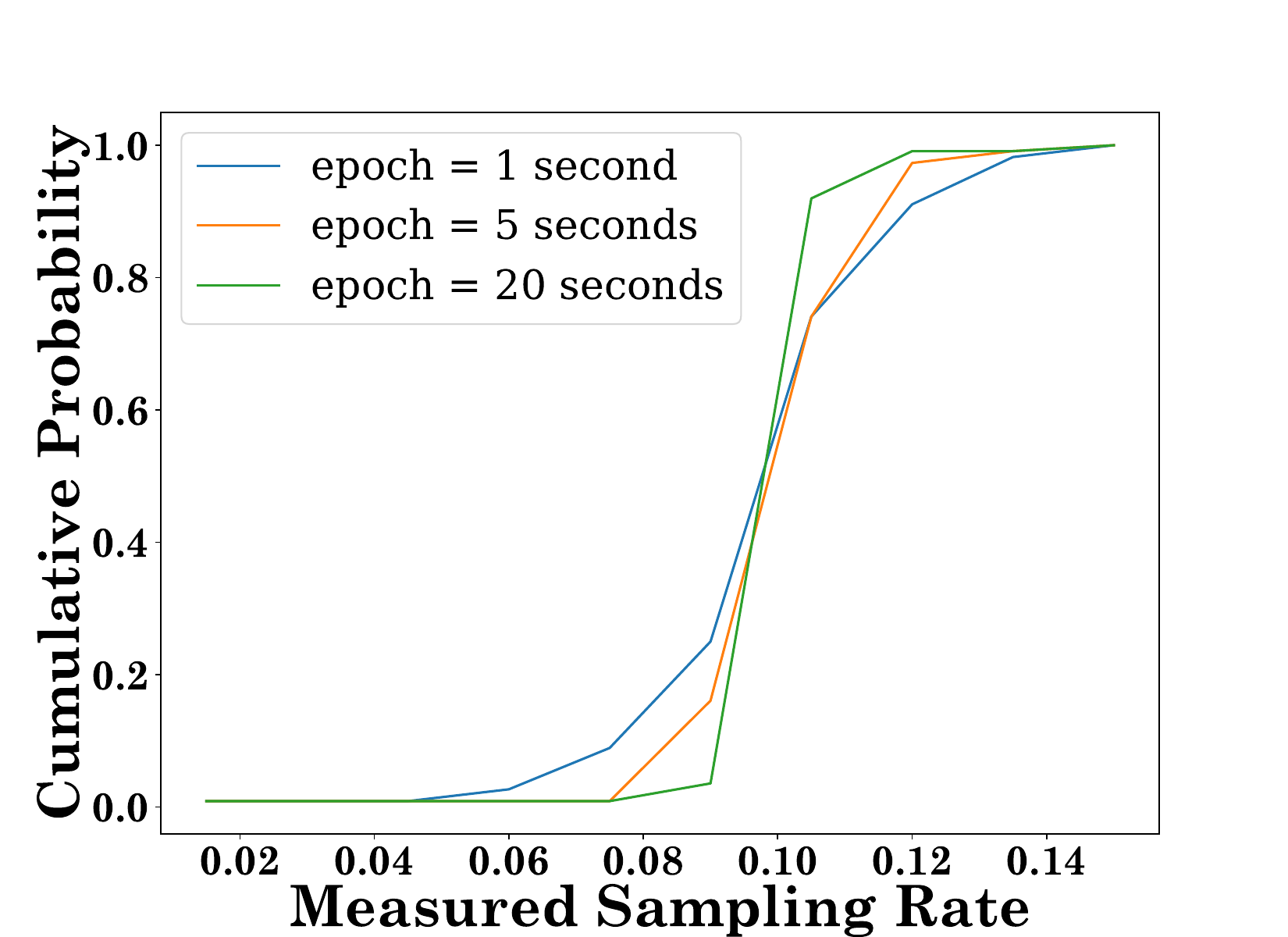}
		\caption{Effect of epoch length on \APX.}
		\label{Simulation_Time_effect}
	\end{minipage}	
	\vspace{-0.5cm}
\end{figure*}

\cat{Implemented Algorithms} In addition to \ILP\ and \APX, we have implemented the following algorithms for comparison:

\begin{itemize}[leftmargin=*]
	\item \textit{Deterministic Sampler (\textbf{DS}):} This algorithm makes sampling decisions solely based on the mean flow rates. 
	
	\item \textit{Deterministic with Headroom (\textbf{DS+2$\boldsymbol{\sigma}$)}:} This algorithm is similar to \DS\ except that it adds a $2\sigma$ headroom to flow rates in order to allow the sampling algorithm to absorb flow rate fluctuations at runtime. It is the approach proposed in~\cite{cohen2018sampling} to deal with dynamic traffic rates. Notice that, for a normal distribution, $95\%$ of samples reside within $2\sigma$ of the mean.
	
	\item \textit{cSamp+$\boldsymbol{\epsilon}$ (\textbf{cS+$\boldsymbol{\epsilon}$)}:} This algorithm is presented in~\cite{sekar2008csamp}. It assumes that the fluctuation in traffic rates around their means is bounded by some $\epsilon > 0$. It then scales the solution computed by the \DS\ algorithm to account for the rate fluctuations that can happen in the worst case.
\end{itemize}
At a high level, \APX\ also reserves some headroom on each switch to deal with rate fluctuations. However, in \APX, the amount of headroom depends on the actual variability (variance) of flow rates, while in \DStwo, it is a fixed value. Also, while the headroom in cSamp depends on the actual flow rates, it is reserved based on the worst case fluctuations of rates leading to a consistently more conservative behavior than \APX, as we will show later in this section.

\cat{Performance Metrics}
The main performance metric, referred to as \textit{``Fully Sampled Flows''},  is the number of flows that are fully sampled by the sampling algorithm, \ie\ whose target sampling rate is satisfied by the sampling algorithm. This metric concisely captures the sampling accuracy achieved by an algorithm. However, to provide more insights about the performance of each sampling algorithm, we also report the following metrics: 
\begin{itemize}
\item \textit{Admitted Flows:} The number of flows that are chosen by the algorithm for sampling. The sampled flows may not be sampled at their desired target sampling rate.
\item \textit{Sampling Rate:} The actual sampling rates measured from the simulations. We would like the measured Sampling Rate to be close to the target sampling rate.
\end{itemize}

\subsection{Micro Benchmarks} \label{microbenchmark}
\cat{\APX\ and \ILP\ Comparison} 
Four sampling scenarios were simulated to compare \APX\ and \ILP. 
Each scenario corresponds to a different combination of number of flows and switch sampling capacity as presented in Table~\ref{tab:APXvsILP_admitted}. The flow rates were set at 200 KBps, each with a CoV of 1. 
The performance results are presented in Fig.~\ref{fig:APXvsILP}. Only the results for the first two scenarios are presented, as \ILP\ was unable to solve the other two scenarios within the allocated time (600 seconds). As can be seen, in both scenarios, \ILP\ achieves a better performance than \APX. However, as shown in Table~\ref{tab:APXvsILP_admitted}, the higher performance of \ILP\ comes at a huge cost in terms of the runtime of the algorithm. Specifically, while \ILP\ can be applied only to small-scale scenarios, it takes \APX\ a few milli-seconds to solve each scenario.

\cat{Scalability Analysis}
We measure the time for Gurobi solver to find a solution for \APX\ in different network scenarios. We utilized the NetworkX Python library to generate random scale-free networks, which is implemented based on the methodology described in~\cite{elmokashfi2010scalability}. The switch capacities were set at 100 pps on all switches.
The results are presented in Fig.~\ref{APX_runtime}. While the computation time increases as the number of switches and flows increases, it remains under a second for $100$ and $500$ switches, and only increases to $1.25$ seconds for the large-scale network with $500$ switches. It is worth mentioning that Gurobi failed to resolve the \ISOCP\ problem for all instances within $600$ seconds, while for \ILP\ problem, it managed to solve the smallest instance within that timeframe.

\cat{Sensitivity to Rate Distribution}
In the derivation of \APX, we applied the tail probability of the normal distribution. A natural question is how sensitive \APX\ is to this assumption. To answer this question, we focus on a single switch, as we applied the normal distribution to compute sampling capacity violation of each switch. We setup $20$ flows through the switch, each with an average rate of 1000 pps and a variance of 10,000 pps$^2$, translating to a CoV of 0.1. We generate the actual rate of each flow in the simulations using Normal, Gamma, Uniform, and T distribution.
The goal is to sample all flows going through the switch, indicated by $x_f = 1$ and $\alpha = 1$ at a violation probability of $5\%$. Utilizing the \APX\ algorithm, it was calculated that the switch requires a minimum capacity of $20,735.6$ pps.
In the simulations, we set the sampling capacity of the switch to a very large value, and then measure the actual sampling load on the switch under each traffic distribution.
The results are presented in Fig.~\ref{Distribution_analysis}. We can see that the median switch load is close to 20,000 packets, with an inter-quartile range between 19500 and 20500 packets, \textit{regardless of the traffic distribution}. Specifically, we found that in only about $5\%$ of all the simulation runs for each distribution, the switch's capacity was violated, which was remarkably close to the desired $5\%$ violation probability. The reason for this is that the violation probability is computed over the sum of flow rates for each switch. Based on the central limit theorem, the sum approaches normal distribution as the number of flows becomes large.

\cat{Effect of Epoch Length}
To assess the effect of epoch length on measured sampling rates, we used Abilene topology with flow rates uniformly set at 200 KBps and a CoV of 1. The capacity of all switches was configured to 100 pps. Out of 100 flows chosen randomly, \APX\ admits 22 of these flows.
The results are depicted in Fig.~\ref{Simulation_Time_effect}. We observe that by increasing the epoch length, the measured sampling rates converge to 0.1, which is the desired target rate.  
However, longer epochs have the downside of reduced responsiveness to new sampling requests that arrive mid-epoch. This presents a trade-off that can be set by network operators depending on their goals. 
\begin{figure*}[t]
	\setkeys{Gin}{width=1.15\linewidth}
	\begin{subfigure}[t]{0.31\textwidth}
		\centering  \includegraphics[]{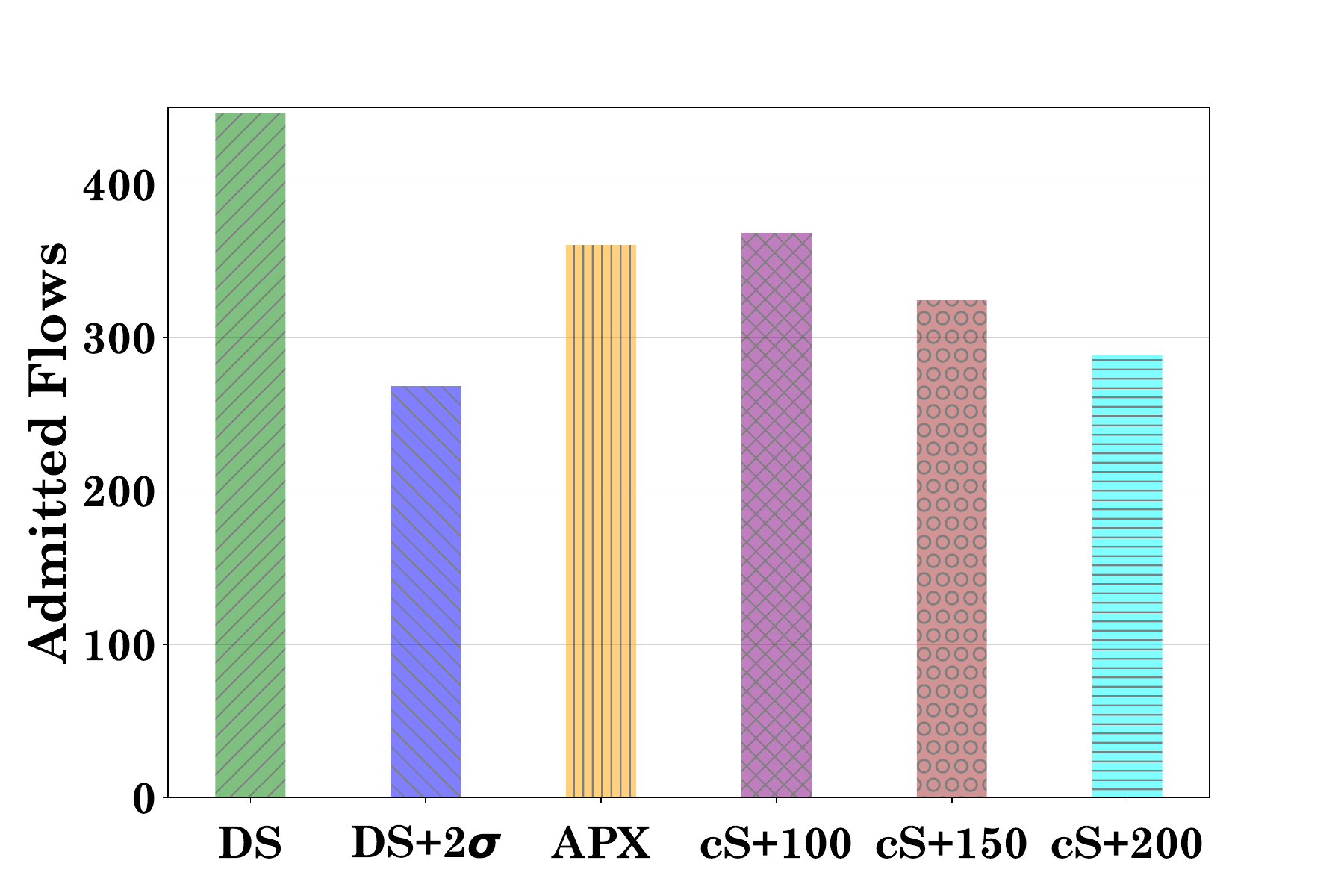}%
		\caption{Total number of admitted flows.}
		\label{comparison_ModelA}
	\end{subfigure}\hfill
	\begin{subfigure}[t]{0.31\textwidth}
		\centering \includegraphics[]{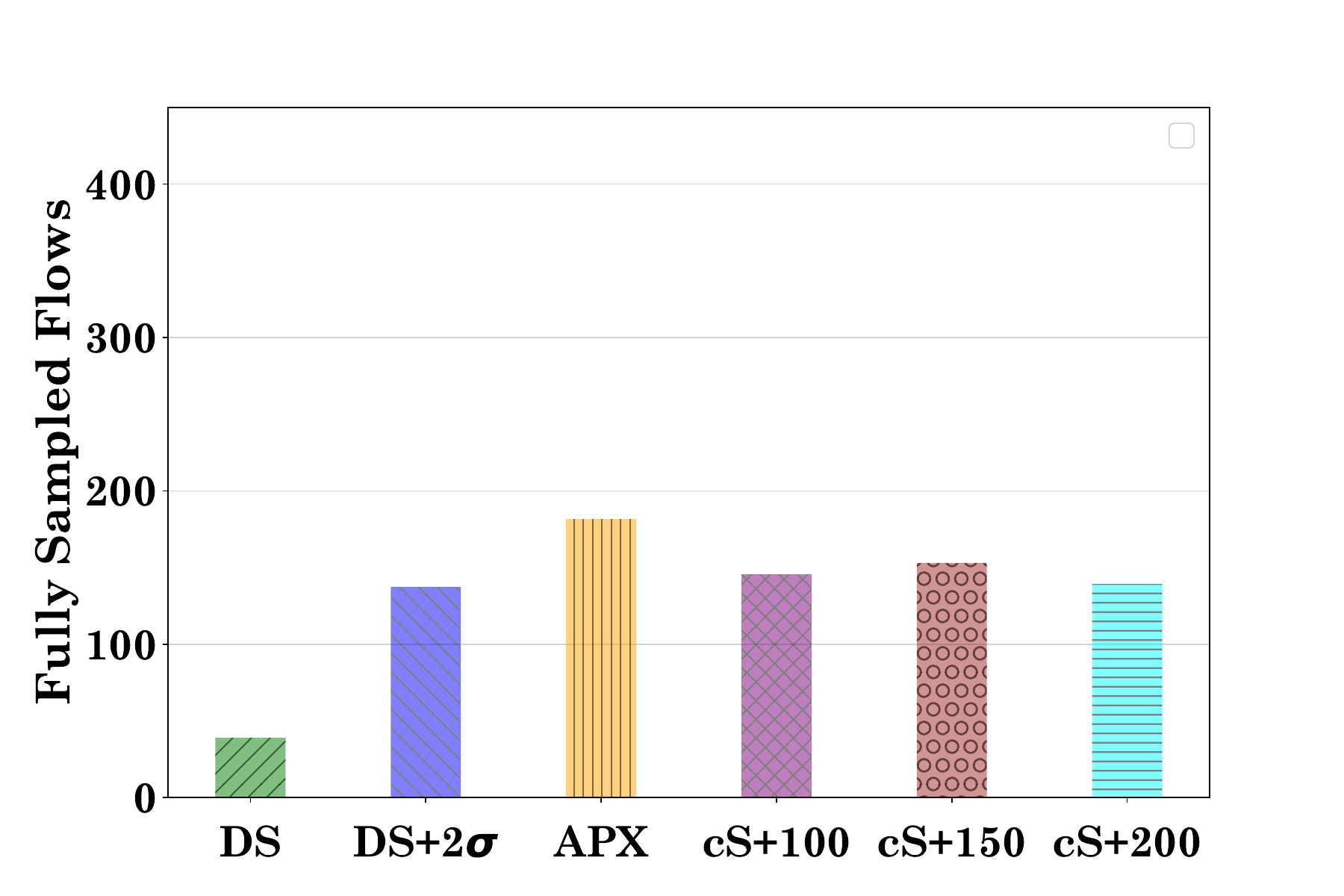}%
		\caption{Total number of fully sampled flows.}
		\label{comparison_ModelB}
	\end{subfigure}\hfill
	\begin{subfigure}[t]{0.31\textwidth}
		\centering \includegraphics[]{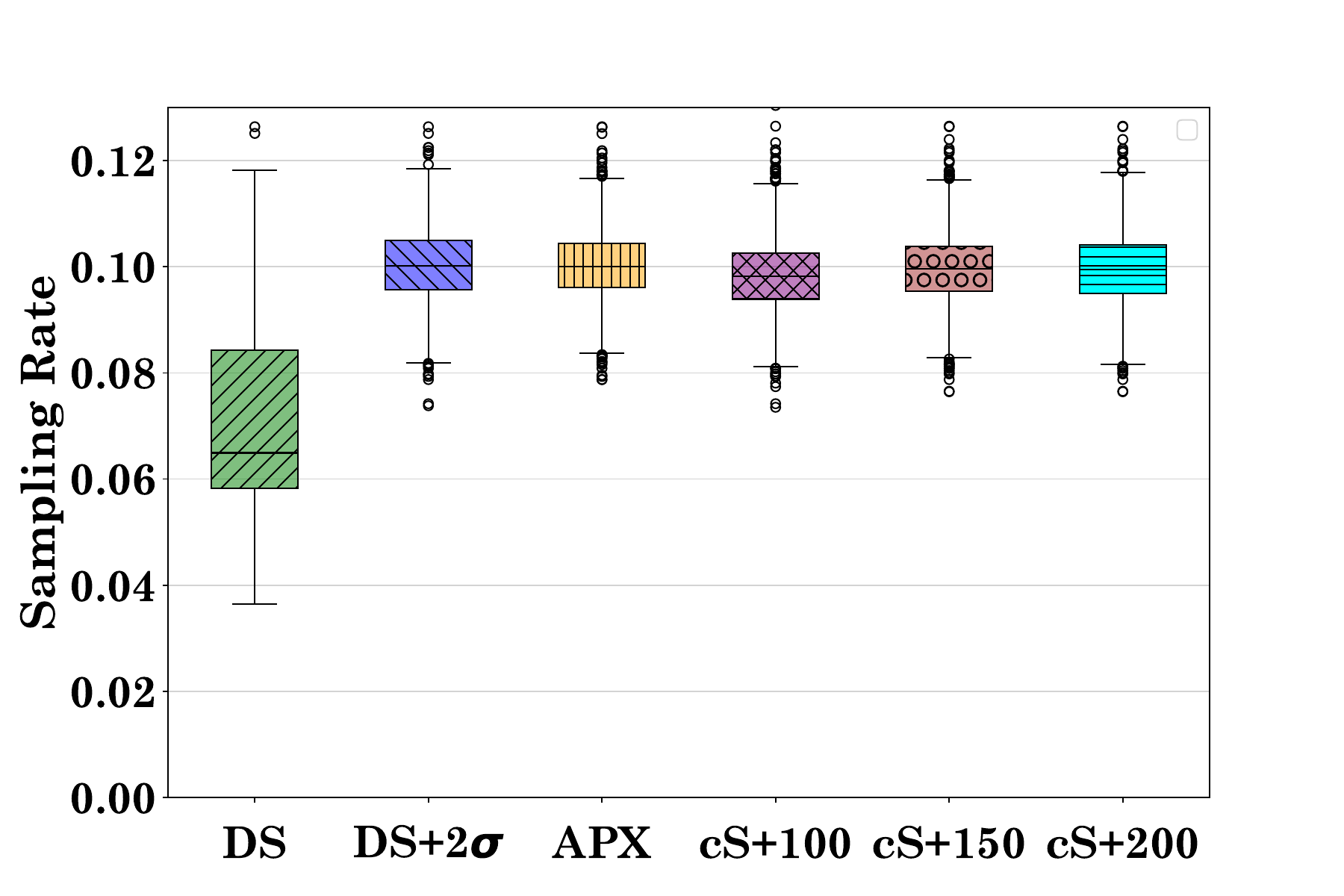}%
		\caption{Measured sampling rates (excluding zero rates).}
		\label{comparison_ModelC}
	\end{subfigure}
	\caption{Model-driven ns-3 simulations results.\label{comparison_Model}}%
	\vspace{-0.5cm}
\end{figure*}

\subsection{Model-Driven Simulations}\label{model_driven}
\cat{Traffic Generation}
Since the Abilene topology is used, a total of 550 flows ($110\times5$) are available for sampling over 5 epochs. Out of these, 446 flows (about $80\%$) were chosen across all 5 epochs. 
We assign flow rates by randomly selecting their mean from the set $\{200,300,500\}$ KBps. Additionally, for each flow, with the probability of $30\%$ and $70\%$, its CoV is set to $0.2$ (low variability) or $2$ (high variability), respectively.
%
%
The flow rates are then sampled from a truncated normal distribution based on the assigned mean and variance. During each experiment, the flow rates change every 0.1 seconds. 
The switch sampling capacities are set to 400 pps.

\cat{Results and Discussion}
The results are depicted in Fig.~\ref{comparison_Model}. Specifically, Fig.~\ref{comparison_ModelA} shows Admitted Flows for each algorithm (from 446 flows). As expected, \DS\ is the most aggressive algorithm due to its lack of consideration for flow rate fluctuations. 
Fig.~\ref{comparison_ModelB} presents the results for Fully Sampled Flows with different algorithms. We see that while \DS\ was the most aggressive in admitting flows, it achieves the lowest Fully Sampled Flows among all the algorithms. For cSamp, we had to decide about the parameter $\epsilon$. We experimented with a range of values, but have only plotted \csamphundred, \csamphundredfifty, and \csamptwohundred, as we found that other values make the algorithm too aggressive or conservative leading to lower performance. Nevertheless, we observe that \APX\ can fully sample more flows that even the best version of cSamp, namely \csamphundredfifty. Not surprisingly, we also find that \DStwo\ falls somewhere in the middle of the pack, but clearly less efficient than \APX. Finally, Fig.~\ref{comparison_ModelC} presents the Sampling Rate of each algorithm, where it can be seen that except DS and \csamphundred, all other algorithms can actually achieve sampling rates that are very close to the target sampling rate. In fact, while there are some variations in measured sampling rates for these algorithms, the median sampling rates closely match the $0.10$ target rate.

\begin{figure*}
	\setkeys{Gin}{width=1.15\linewidth}
	\begin{subfigure}[t]{0.31\textwidth}
		\centering  \includegraphics[]{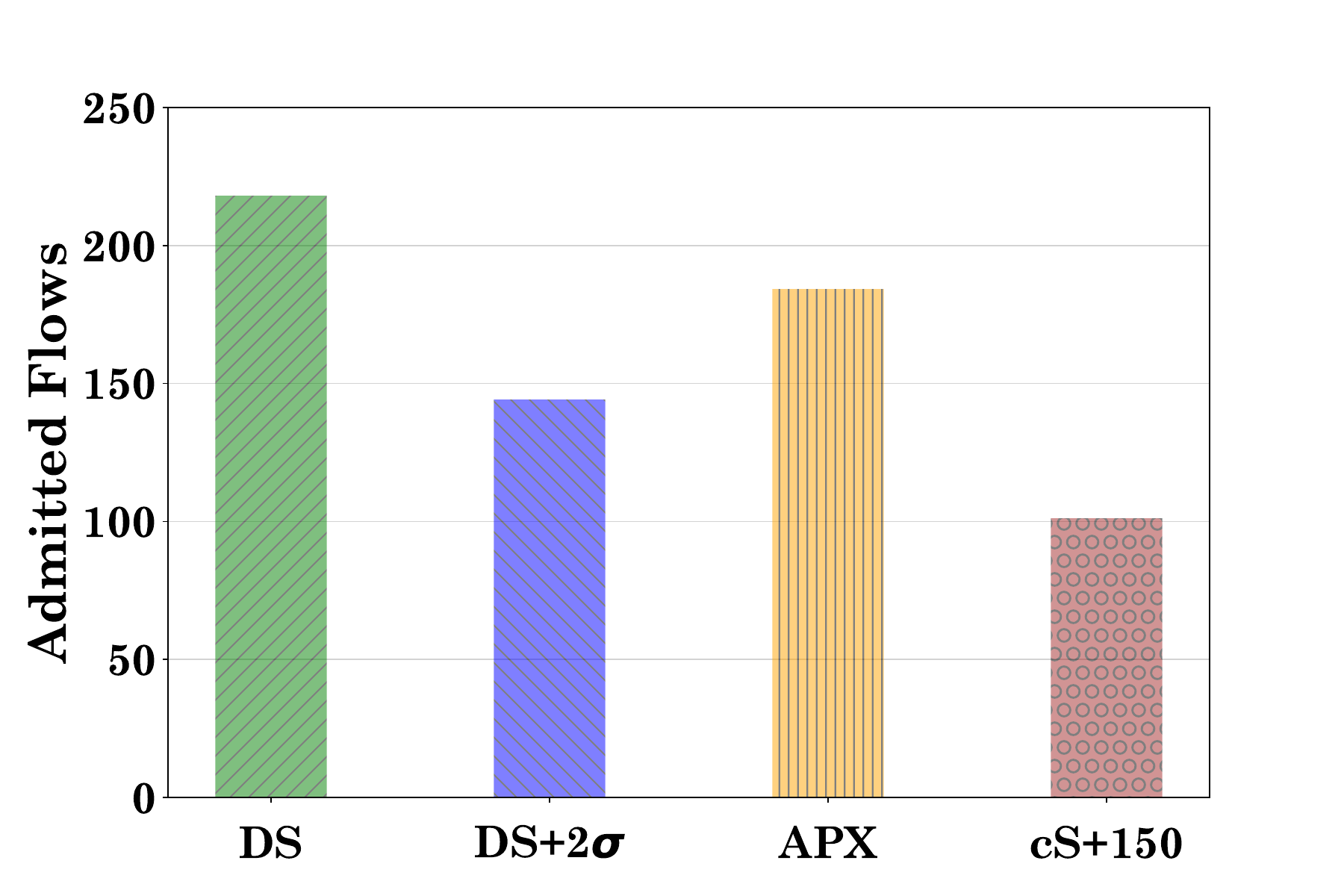} %
		\caption{Total number of admitted flows.}
		\label{comparison_TraceA}
	\end{subfigure}\hfill
	\begin{subfigure}[t]{0.31\textwidth}
		\centering
		\includegraphics[]{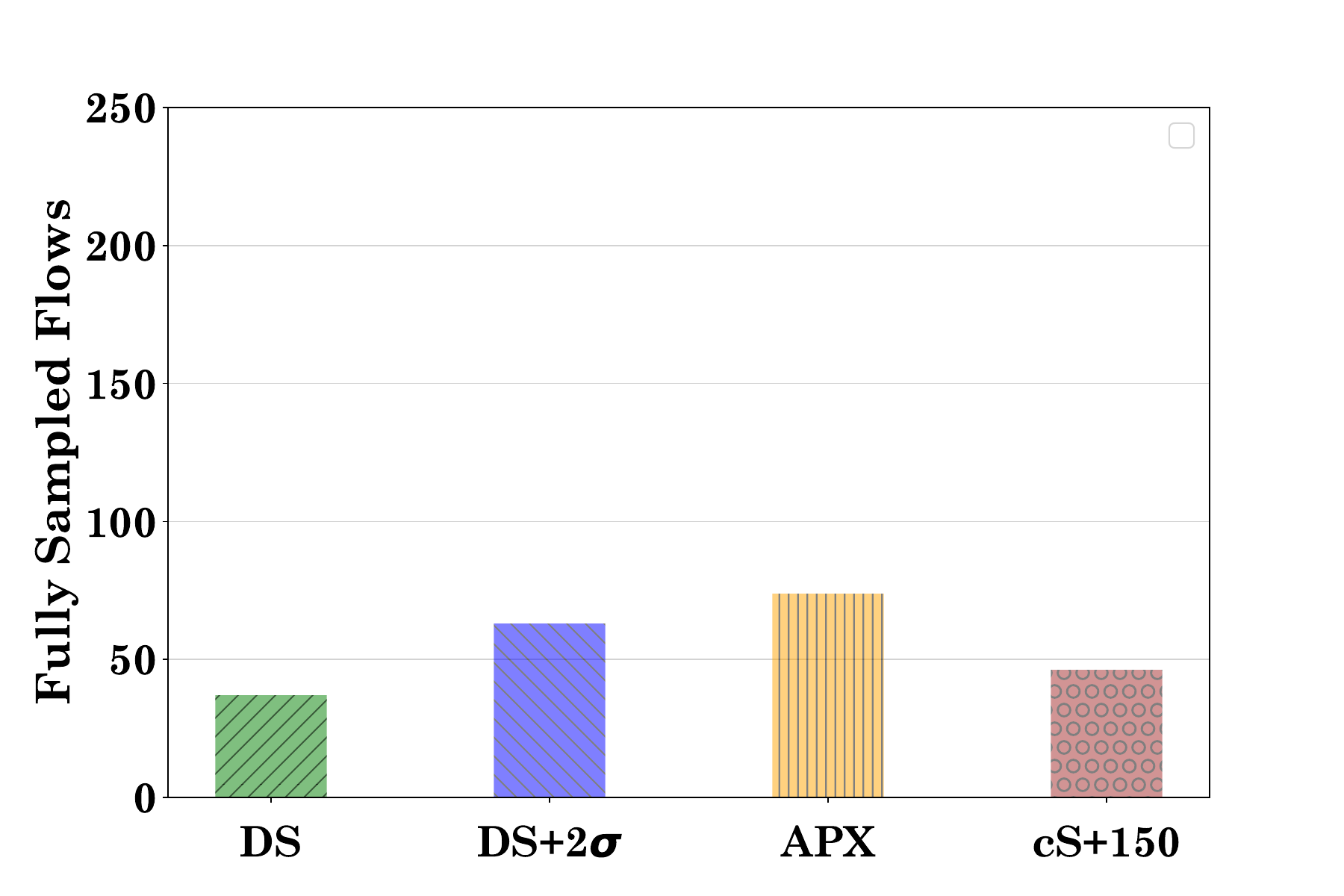}%
		\caption{Total number of fully sampled flows.}
		\label{comparison_TraceB}
	\end{subfigure}\hfill
	\begin{subfigure}[t]{0.31\textwidth}
		\centering \includegraphics[]{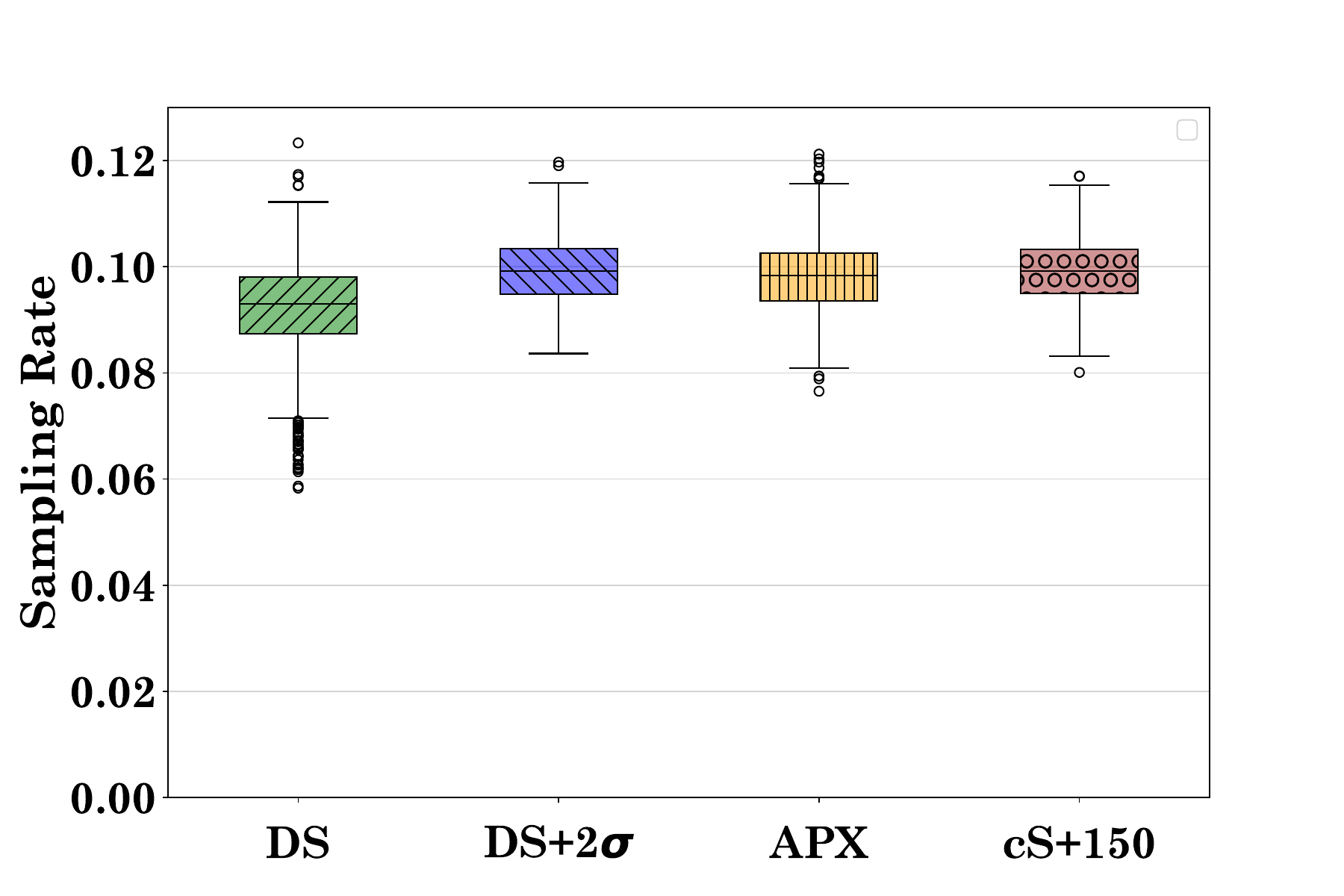}%
		\caption{Measured sampling rates (excluding zero rates).}
		\label{comparison_TraceC}
	\end{subfigure}
	\caption{Trace-driven ns-3 simulation results.}%
	\label{comparison_Trace}
	\vspace{-0.3cm}
\end{figure*}
\subsection{Trace-Driven Simulations}\label{trace_driven}
In this experiment, for comparison between \APX\ and cSamp, we selected \csamphundredfifty, as it was the top-performing cSamp version in the model-driven simulations.

\cat{Traffic Generation} In these simulations, we utilized the actual traffic data obtained from an ISP backbone link\cite{MAWI} to generate rates for various flows. Each traffic trace consists of detailed packet-level information, capturing traffic over the ISP link for a duration of 15 minutes on a specific day. 
The traffic rate for each flow was changed every 100 milliseconds throughout the simulation, as in the model-driven simulations. In order to speed up the simulations, all traffic rates from the original files were divided by 100. This scaling down does not alter the CoV in comparison to the original data.
Similar to the model-driven simulations, the Abilene topology was used, and there were a total of 550 flows, out which 443 flows (about $80\%$) were randomly chosen for sampling. The sampling capacity of each switch was set to 200 pps.

\cat{Results and Discussion}
The results are shown in Fig.~\ref{comparison_Trace}.
A few observations are in order. First, we observe that \APX\ outperforms all the other algorithms while still achieving a median sampling rate that matches the target sampling rate. Second, while the relative performance of the \DS, \DStwo, and \APX\ algorithms is similar to what was observed in model-driven simulations, \csamphundredfifty\ performs poorly even compared to \DStwo. The reason is that \csamphundredfifty\ considers the worst-case rate fluctuations, which makes it very conservative under the ISP trace, which is smoother compared to model-driven traces. This is indeed a major problem with both \DStwo\ and cSamp in networks with dynamic traffic. The parameters of these algorithms are tuned once at the deployment time based on the expected traffic characteristics, but when inevitably traffic characteristics change over time, they achieve sub-optimal performance. In contrast, \APX\ has no parameters to be tuned for each specific type of traffic, and instead adapts to traffic changes by design, leading to higher performance than the other algorithms.

\section{Related Works}
\label{s:related}
Traffic sampling is a well-studied area. In the following, we briefly review some representative works with a focus on per-flow sampling solutions.

\cat{Per-Port Sampling} 
NetFlow~\cite{claise2004cisco} and sFlow~\cite{sflow} are widely used for per-port sampling. In per-port sampling, a switch samples a fraction of packets that pass through one or more of its ports regardless of which flow these packets belong to. Several works try to improve the performance of NetFlow and sFlow by more intelligently distributing sampling load on network switches. For instance, the works~\cite{xu2019lightweight} and~\cite{sadrhaghighi2021softtap} propose sampling solutions that aim at minimizing the sampling load on network switches by coordinating sampling across switches. However, since per-port sampling solutions are based on a fixed sampling rate for each switch port, they are biased toward flows that send at a higher rate, which results in low accuracy for many management tasks that require a minimum per-flow sampling rate~\cite{mai2006sampled,carela2011analysis,ha2016suspicious}.

\cat{Uncoordinated Per-Flow Sampling} 
The works in this category consider per-flow sampling, but the sampling decisions are made on a per-switch basis without any network-wide coordination. For instance, the works~\cite{estan2004building} and~\cite{kompella2005power} propose algorithms for adjusting the sampling rate on each switch in response to variations in traffic rate. Other works in this category focus on providing better traffic estimates from the sampled packets~\cite{duffield2001charging,hohn2003inverting} or reducing the communication overhead~\cite{du2022self}. These works, however, do not consider sampling coordination among network switches, and as such, are less efficient compared to coordinated sampling solutions. The reason is that when switches sample flows without coordination, it is possible that some packets are sampled redundantly by multiple switches. Given the limited sampling capacity of switches, the result is a lower effective sampling rate for each flow, which translates into lower flow visibility in the network.

\cat{Coordinated Per-Flow Sampling} The works in this category leverage a centralized controller to coordinate sampling responsibilities across network switches. The works~\cite{sekar2008csamp,sekar2010revisiting,chang2011leisure,cohen2018sampling} distribute sampling load among different switches. These works, however, assume that the sampling controller knows the set of flows and their sending rates in advance, which is not realistic.
The work~\cite{sadrhaghighi2022flowshark} considers coordinated sampling in an online setting, where the controller does not have any information about the flows that will arrive in the system for sampling in the future. They also separate sampling decisions for short and long flows, by having a fixed sampling rate for short flows and only performing coordinated sampling for long flows.
However, all the aforementioned solutions assume fixed flow rates during the sampling period, which may lead to under-sampling, resulting in lower accuracy in some network management tasks such as anomaly detection. Several works consider flow rate fluctuations in sampling to increase the accuracy of per-flow sampling. For instance, OpenMeasure~\cite{liu2016openmeasure} predicts the size of flows based on previously collected flow samples. While their method considers variations in flow sizes across different measurement intervals, it does not account for the instantaneous changes in flow rates within the same measurement interval. 
Perhaps the work closest to ours is cSamp~\cite{sekar2008csamp}, which considers coordinated network-wide sampling. While the cSamp sampling algorithm is designed assuming fixed flow rates, it has a post-processing mechanism to make the solution robust against estimation errors in flow rates. Specifically, the sampling solution computed for fixed flow rates is \textit{scaled} to account for the maximum deviation in flow rates based on the assumed estimation error. As we show in Section~\ref{s:eval}, this approach is very conservative by essentially considering the worst-case flow rates, which means that it may fail to fully utilize the available network sampling capacity. A somewhat similar approach is considered in~\cite{cohen2018sampling}, where the authors scale the switch sampling capacities by a fixed factor to leave some headroom for potential flow rate fluctuations. However, a fixed headroom is not very efficient since it does not account for dynamic changes in network traffic, as demonstrated in our evaluations.

In contrast to the above works, our approach is to explicitly include flow rate fluctuations in the design of the sampling controller. Rather than assuming the worst-case flow rates, we consider a probabilistic approach in which we directly bound the probability that any switch is overloaded (its sampling capacity is violated) due to fluctuation in flow rates.

\section{Conclusion}
\label{s:conc}
In this paper, we presented the design and evaluation of a sampling system called \SysName, capable of handling dynamic flow rates. We formulated the problem of network-wide coordinated sampling with dynamic flow rates as an ISOCP, which proved infeasible for even small-scale networks. To address the computational complexity of larger networks, we created two alternative reformulations: one transformation and one approximation. By evaluating their runtime in various scenarios, we demonstrated that the transformation works well for small-scale networks, while the approximation is more effective for larger-scale networks. Through both model-driven and trace-driven ns-3 simulations, we showed that our proposed system outperforms existing sampling systems in achieving a higher percentage of fully sampled flows. A potential future direction is to design an online algorithm to solve the problem, where sampling requests are processed instantaneously instead of being batched over sampling epochs.

\bibliographystyle{IEEEtran}
\bibliography{bib/IEEEabrv,bib/main} 

\begin{thebibliography}{10}
\providecommand{\url}[1]{#1}
\csname url@samestyle\endcsname
\providecommand{\newblock}{\relax}
\providecommand{\bibinfo}[2]{#2}
\providecommand{\BIBentrySTDinterwordspacing}{\spaceskip=0pt\relax}
\providecommand{\BIBentryALTinterwordstretchfactor}{4}
\providecommand{\BIBentryALTinterwordspacing}{\spaceskip=\fontdimen2\font plus
\BIBentryALTinterwordstretchfactor\fontdimen3\font minus
  \fontdimen4\font\relax}
\providecommand{\BIBforeignlanguage}[2]{{%
\expandafter\ifx\csname l@#1\endcsname\relax
\typeout{** WARNING: IEEEtran.bst: No hyphenation pattern has been}%
\typeout{** loaded for the language `#1'. Using the pattern for}%
\typeout{** the default language instead.}%
\else
\language=\csname l@#1\endcsname
\fi
#2}}
\providecommand{\BIBdecl}{\relax}
\BIBdecl

\bibitem{tsai2018network}
P.-W. Tsai, C.-W. Tsai, C.-W. Hsu, and C.-S. Yang, ``Network monitoring in
  software-defined networking: A review,'' \emph{IEEE Systems Journal},
  vol.~12, no.~4, 2018.

\bibitem{li2011scan}
T.~Li, S.~Chen, W.~Luo, and M.~Zhang, ``Scan detection in high-speed networks
  based on optimal dynamic bit sharing,'' in \emph{Proc. IEEE INFOCOM}, 2011.

\bibitem{tian2018sdn}
Y.~Tian, W.~Chen, and C.-T. Lea, ``An {SDN}-based traffic matrix estimation
  framework,'' \emph{{IEEE} Trans. Netw. Service Manag.}, vol.~15, no.~4, 2018.

\bibitem{duffield2001charging}
N.~Duffield, C.~Lund, and M.~Thorup, ``Charging from sampled network usage,''
  in \emph{Proc. ACM SIGCOMM, Workshop on Internet Measurement}, 2001.

\bibitem{zhang2013unsupervised}
J.~Zhang, Y.~Xiang, W.~Zhou, and Y.~Wang, ``Unsupervised traffic classification
  using flow statistical properties and {IP} packet payload,'' \emph{Journal of
  Computer and System Sciences}, vol.~79, no.~5, 2013.

\bibitem{cohen2018sampling}
R.~Cohen and E.~Moroshko, ``Sampling-on-demand in {SDN},'' \emph{{IEEE/ACM}
  Trans. Netw.}, vol.~26, no.~6, 2018.

\bibitem{li2016flowradar}
Y.~Li, R.~Miao, C.~Kim, and M.~Yu, ``{FlowRadar}: A better {NetFlow} for data
  centers,'' in \emph{Proc. USENIX NSDI}, 2016.

\bibitem{suh2014opensample}
J.~Suh, T.~T. Kwon, C.~Dixon, W.~Felter, and J.~Carter, ``Opensample: A
  low-latency, sampling-based measurement platform for commodity {SDN},'' in
  \emph{Proc. IEEE ICDCS}, 2014.

\bibitem{claise2004cisco}
B.~Claise, ``Cisco systems {NetFlow} services export version 9,'' Tech. Rep.,
  2004.

\bibitem{sflow}
\BIBentryALTinterwordspacing
{sFlow} {"Making the network visible"}. [Online]. Available:
  \url{https://sflow.org/}
\BIBentrySTDinterwordspacing

\bibitem{cantieni2006reformulating}
G.~R. Cantieni, G.~Iannaccone, C.~Barakat, C.~Diot, and P.~Thiran,
  ``Reformulating the monitor placement problem: Optimal network-wide
  sampling,'' in \emph{Proc. ACM CoNEXT}, 2006.

\bibitem{suarez2017towards}
J.~Su{\'a}rez-Varela and P.~Barlet-Ros, ``Towards a {NetFlow} implementation
  for {OpenFlow} software-defined networks,'' in \emph{Proc. IEEE ITC}, 2017.

\bibitem{shirali2013flexam}
S.~Shirali-Shahreza and Y.~Ganjali, ``Flexam: Flexible sampling extension for
  monitoring and security applications in openflow,'' in \emph{Proc. ACM
  SIGCOMM, Workshop on Hot Topics in SDN}, 2013.

\bibitem{raspall2012efficient}
F.~Raspall, ``Efficient packet sampling for accurate traffic measurements,''
  \emph{Computer Networks}, vol.~56, no.~6, 2012.

\bibitem{sekar2008csamp}
V.~Sekar, M.~K. Reiter, W.~Willinger, H.~Zhang, R.~R. Kompella, and D.~G.
  Andersen, ``{cSamp}: A system for network-wide flow monitoring,'' in
  \emph{Proc. USENIX NSDI}, 2008.

\bibitem{sekar2010revisiting}
V.~Sekar, M.~K. Reiter, and H.~Zhang, ``Revisiting the case for a minimalist
  approach for network flow monitoring,'' in \emph{Proc. ACM SIGCOMM}, 2010.

\bibitem{sadrhaghighi2022flowshark}
S.~Sadrhaghighi, M.~Dolati, M.~Ghaderi, and A.~Khonsari, ``{FlowShark}:
  Sampling for high flow visibility in {SDNs},'' in \emph{Proc. IEEE INFOCOM},
  2022.

\bibitem{zhou2006traffic}
B.~Zhou, D.~He, and Z.~Sun, ``Traffic predictability based on {ARIMA/GARCH}
  model,'' in \emph{Proc. NGI}, 2006.

\bibitem{andreoletti2019network}
D.~Andreoletti, S.~Troia, F.~Musumeci, S.~Giordano, G.~Maier, and M.~Tornatore,
  ``Network traffic prediction based on diffusion convolutional recurrent
  neural networks,'' in \emph{Proc. IEEE INFOCOM, Workshop on Network
  Intelligence}, 2019.

\bibitem{sekar2010coordinated}
V.~Sekar, A.~Gupta, M.~K. Reiter, and H.~Zhang, ``Coordinated sampling sans
  origin-destination identifiers: algorithms and analysis,'' in \emph{Proc.
  COMSNETS}, 2010.

\bibitem{chernoff1952measure}
H.~Chernoff, ``A measure of asymptotic efficiency for tests of a hypothesis
  based on the sum of observations,'' \emph{The Annals of Mathematical
  Statistics}, pp. 493--507, 1952.

\bibitem{GurobiOptimization}
\BIBentryALTinterwordspacing
{Gurobi Optimization}. [Online]. Available: \url{http://www.gurobi.com/}
\BIBentrySTDinterwordspacing

\bibitem{MAWI}
\BIBentryALTinterwordspacing
{MAWI} {"MAWI Working Group Traffic Archive"}. [Online]. Available:
  \url{http://mawi.wide.ad.jp/mawi/}
\BIBentrySTDinterwordspacing

\bibitem{elmokashfi2010scalability}
A.~Elmokashfi, A.~Kvalbein, and C.~Dovrolis, ``On the scalability of {BGP}: The
  role of topology growth,'' \emph{{IEEE} J. Sel. Areas Commun.}, vol.~28,
  no.~8, 2010.

\bibitem{xu2019lightweight}
H.~Xu, S.~Chen, Q.~Ma, and L.~Huang, ``Lightweight flow distribution for
  collaborative traffic measurement in software defined networks,'' in
  \emph{Proc. IEEE INFOCOM}, 2019.

\bibitem{sadrhaghighi2021softtap}
S.~Sadrhaghighi, M.~Dolati, M.~Ghaderi, and A.~Khonsari, ``{SoftTap}: A
  software-defined {TAP} via switch-based traffic mirroring,'' in \emph{Proc.
  IEEE NetSoft}, 2021.

\bibitem{mai2006sampled}
J.~Mai, C.-N. Chuah, A.~Sridharan, T.~Ye, and H.~Zang, ``Is sampled data
  sufficient for anomaly detection?'' in \emph{Proc. ACM SIGCOMM}, 2006.

\bibitem{carela2011analysis}
V.~Carela-Espa{\~n}ol, P.~Barlet-Ros, A.~Cabellos-Aparicio, and
  J.~Sol{\'e}-Pareta, ``Analysis of the impact of sampling on {NetFlow} traffic
  classification,'' \emph{Computer Networks}, vol.~55, no.~5, 2011.

\bibitem{ha2016suspicious}
T.~Ha, S.~Kim, N.~An, J.~Narantuya, C.~Jeong, J.~Kim, and H.~Lim, ``Suspicious
  traffic sampling for intrusion detection in software-defined networks,''
  \emph{Computer Networks}, vol. 109, 2016.

\bibitem{estan2004building}
C.~Estan, K.~Keys, D.~Moore, and G.~Varghese, ``Building a better {NetFlow},''
  \emph{ACM SIGCOMM Computer Communication Review}, vol.~34, no.~4, 2004.

\bibitem{kompella2005power}
R.~R. Kompella and C.~Estan, ``The power of slicing in {Internet} flow
  measurement,'' in \emph{Proc. ACM SIGCOMM}, 2005.

\bibitem{hohn2003inverting}
N.~Hohn and D.~Veitch, ``Inverting sampled traffic,'' in \emph{Proc. ACM
  SIGCOMM}, 2003.

\bibitem{du2022self}
Y.~Du, H.~Huang, Y.-E. Sun, S.~Chen, G.~Gao, and X.~Wu, ``Self-adaptive
  sampling based per-flow traffic measurement,'' \emph{{IEEE/ACM} Trans.
  Netw.}, 2022.

\bibitem{chang2011leisure}
C.-W. Chang, G.~Huang, B.~Lin, and C.-N. Chuah, ``Leisure: A framework for
  load-balanced network-wide traffic measurement,'' in \emph{Proc. ACM/IEEE
  ANCS}, 2011.

\bibitem{liu2016openmeasure}
C.~Liu, A.~Malboubi, and C.-N. Chuah, ``{OpenMeasure}: Adaptive flow
  measurement \& inference with online learning in {SDN},'' in \emph{Proc. IEEE
  INFOCOM WKSHPS}, 2016.

\end{thebibliography}

\end{document}